\documentclass[letterpaper,10 pt, conference]{ieeeconf}
\IEEEoverridecommandlockouts                               
\usepackage{fancyhdr}
\usepackage{xspace}
\usepackage[font={small}]{caption}
\usepackage{color}
\usepackage{setspace}
\usepackage{enumerate}
\usepackage{graphics}
\usepackage{cite}
\usepackage{dsfont}
\usepackage{latexsym}
\usepackage{amsmath}
\usepackage{amssymb}
\usepackage{amsfonts}
\usepackage{amsthm}
\usepackage{times}
\usepackage{sgame}
\usepackage{color}
\usepackage{verbatim}
\usepackage{xcolor}
\usepackage{pgfplots}
\usepackage{epsfig} 
\usepackage{subcaption}
\usepackage{tabularx}
\newcommand{\ignore}[1]{}
\usepackage[hidelinks]{hyperref}

\let\labelindent\relax
\usepackage{enumitem}
\newlist{inparaenum}{enumerate}{2}
\setlist[inparaenum]{nosep}
\setlist[inparaenum,1]{label=\bfseries\arabic*.}
\setlist[inparaenum,2]{label=\arabic{inparaenumi}\emph{\alph*})}

\newcommand{\be}{\begin{equation}}
\newcommand{\ee}{\end{equation}}

\theoremstyle{plain}
\newtheorem{theorem}{Theorem}[section]
\newtheorem{corollary}{Corollary}
\newtheorem{proposition}{Proposition}[section]

\def\sf{S}

\normalsize\title{Optimizing Preventive and Reactive Defense Resource Allocation with Uncertain Sensor Signals}

\author{
	Faezeh Shojaeighadikolaei, Shouhuai Xu, Keith Paarporn
	\thanks{F. Shojaeighadikolaei, S. Xu, and K. Paarporn are with the Department of Computer Science, University of Colorado Colorado Springs. This work was supported in part by NSF grant \#ECCS-2346791. Contact: \texttt{ \{fshojaei,sxu,kpaarpor\}@uccs.edu}.
	}
}

\begin{document}

\maketitle

\begin{abstract}
Cyber attacks continue to be a cause of concern despite advances in cyber defense techniques. Although cyber attacks cannot be fully prevented, standard decision-making frameworks typically focus on how to prevent them from succeeding,  without considering the cost of cleaning up the damages incurred by successful attacks. This motivates us to investigate a new resource allocation problem formulated in this paper: The defender must decide how to split its investment between preventive defenses, which aim to harden nodes from attacks, and reactive defenses, which aim to quickly clean up the compromised nodes. This encounters a challenge imposed by the uncertainty associated with the observation, or sensor signal, whether a node is truly compromised or not; this uncertainty is real because attack detectors are not perfect.
We investigate how the quality of sensor signals impacts the defender's strategic investment in the two types of defense, and ultimately the level of security that can be achieved. In particular, we show that the optimal investment in preventive resources increases, and thus reactive resource investment decreases, with higher sensor quality. We also show that the defender's performance improvement, relative to a baseline of no sensors employed, is maximal when the attacker can only achieve low attack success probabilities.

\end{abstract}

\section{Introduction}
It has proven impossible to prevent cyber attacks from succeeding
\cite{Cohen1988} even without considering human factors, which usher in another dimension of vulnerabilities  \cite{Manshaei2013} for which there are no adequate defenses \cite{XuPIEEE2024}. 
This highlights the importance of collectively employing multiple kinds of defenses \cite{Zhu2024, Linkov2019, Xu2019, Zhu2015,Kazeminajafabadi,10896861, Behnam2022}. 
Recently, there has been much attention on differentiating between a system's robustness (ability to withstand attacks) and resilience (ability to recover from successful attacks). In this paper, we investigate their interplay via two kinds of cyber defenses:  {\em preventive defenses}, such as access control mechanisms that contributes to system robustness; and {\em reactive defenses}, such as clean-up processes.

One important research question  is: {\em How should a defender optimally invest limited resources towards both  preventive and reactive defenses?} At one extreme, total investment in preventive defenses minimizes chances for potential attacks to succeed, but leaves no resources to recover from compromises in the event that they occur. On the other extreme, total investment in reactive defense maximizes its capability to recover from compromises, but leaves systems completely vulnerable to compromise. These tradeoffs signify an emerging paradigm in cyber defense decision-making for security that warrants systematic investigation, given that previous decision-making frameworks have primarily focused on preventive defenses \cite{Linkov2019,Paarporn2023,Iliaev2023, Chen2020,Miehling2018}.

The collective use of preventive and reactive defenses has been investigated in recent works \cite{Linkov2019,Paarporn2023,Iliaev2023, Chen2020,Miehling2018}.  For example, \cite{Chen2020} considers strategic attacker-defender interactions in dynamic multi-stage settings, where preventive design, adversarial attacks, and post-attack recovery are modeled and analyzed via subgame perfect equilibria. Closest to the present paper are studies on the allocation of the two types of resources over multiple nodes \cite{Paarporn2023}, which derives the Nash equilibrium under the assumption that the defender does not receive any signals regarding which nodes are compromised. Another thread of research considers preventive and reactive defenses in the context of dynamic computer malware spreading in complex networks \cite{WangSRDS03,XuTAAS2012,XuIEEETNSE2018,XuIEEEACMToN2019,Han2020}.

The present paper is specifically interested in studying how the balance between investing in these two types of defenses may depend on the quality of the defender's sensing (i.e., attack-detection) capabilities. Here, we consider a scenario where the defender, with a limited budget, first allocates resources on preventive defenses to harden nodes (e.g., computers) from cyber attacks. More investment in preventive defenses decreases the probability that attacks will be successful, but leaves fewer resources to reactive defenses. After the attacks are waged, the defender receives noisy signals regarding the nodes' security states (i.e., compromised or not), and then decides how to allocate the remaining resources on reactive defenses. In our model, more investment in reactive defenses decreases a node's expected time to recover from the compromised state. The defender's overall objective is to minimize the costs incurred by nodes in the compromised state for long periods of time.

The informational elements in our model reflect real-world reactive defense mechanisms such as intrusion or attack detectors, which can suffer from inaccuracies and uncertainties. These features have been considered in some cybersecurity models, such as belief-based and stochastic models for investigating security investments under uncertainty~\cite{Mai2022, Zhu2015}. By contrast, we seek to understand the impact that the defender's informational capabilities on promoting overall security, and in doing so, characterize optimal information-aware resource allocation policies.

The present study makes the following contributions. First, we propose a novel defender-centric model for the preventive and reactive defenses trade-off under uncertainty.
We show that the optimal preventive defense investment increases, and thus reactive defense investment decreases, with the quality of the sensors as described by true-positive and false-positive rates (Theorem~\ref{thm:xp_total_vs_pq_n_nodes}). We then quantify the percentage improvement in the expected cost, and how this depends on the sensors' true-positive and false-positive rates (Theorem~\ref{thm:improvement_sensor}). Lastly, we characterize how the attack strength impacts the defender's percent improvement, relative to the baseline scenario where it has uninformative sensors (Theorem~\ref{thm:improvement_gamma}). We conduct numerical simulations to confirm our theoretical results and illustrate the practical impact of sensor quality on optimal budget allocation.

The rest of the paper is organized as follows. Section \ref{sec:model} introduces our model.
Section \ref{sec:result} presents our main theoretical and empirical results. 
Section \ref{sec:con} concludes the paper with a discussion on implications and future directions.

\section{Model}
\label{sec:model}

We study the following two-stage cyber defense resource allocation problem. There are \(n\) nodes, representing computers or devices, which may be compromised by attackers. At each node, the defender can employ: (i) {\em preventive} defense, such as access control, to make it harder for attackers to compromise the node; and/or (ii) {\em reactive} defense, such as intrusion detector or {\em sensor} for succinctness, to detect that a node has been compromised and then take appropriate actions (e.g., cleaning up a compromised node). In practice, sensors are not perfect, with a known true-positive rate and a known false-positive rate.
Suppose the defender has a total budget \(X > 0\) to be spent on preventive defenses and reactive defenses. The attack-defense interaction has two stages.
Figure~\ref{fig:system_diagram} illustrates the model, with model parameters summarized in Table~\ref{tab:notation}. 

\subsection{Stage 1: Preventive defense} 

The defender selects an allocation of preventive resources, \(X_{P,i} \ge 0\) to harden node \(i\), where $1\leq i \leq n$, to reduce its probability of getting compromised, and such that $\sum_{i=1}^n X_{P,i} \leq X$. We assume that an array of attack resources $Y_i$ are allocated to the nodes. These allocations determine the probability that each node becomes compromised: in this work, we define the compromise probability as
\begin{align}
\gamma_i = \frac{Y_i}{Y_i + X_{P,i}}.
\label{eq:prior}
\end{align}
which is decreasing in the investment of preventive resources $X_{P,i}$. The form of \eqref{eq:prior} is known as the Tullock contest success function~\cite{Contest-success-functions-96,Paarporn2023,Gordon2002,Iliaev2023}, which is often adopted to reflect breach probabilities.

\subsection{Stage 2: Reactive defense} 

This stage consists of four components.

\ignore{
the defender allocates a budget \(X_{R,i} \ge 0\), where $1\leq i \leq n$, {\color{red}to clean up node $i$, meaning that  \(X_{R,i} > 0\) implies $S_i=1$ and \(X_{R,i} = 0\) implies $S_i=0$. Note that the defender may result in cleaning up a secure (i.e., non-compromised) node because the sensor has false-positives (i.e., $S_i=1$ does not necessarily mean node $i$ is compromised).}
}

\begin{figure*}[!htbp]
    \centering
    \includegraphics[width=0.9\linewidth]{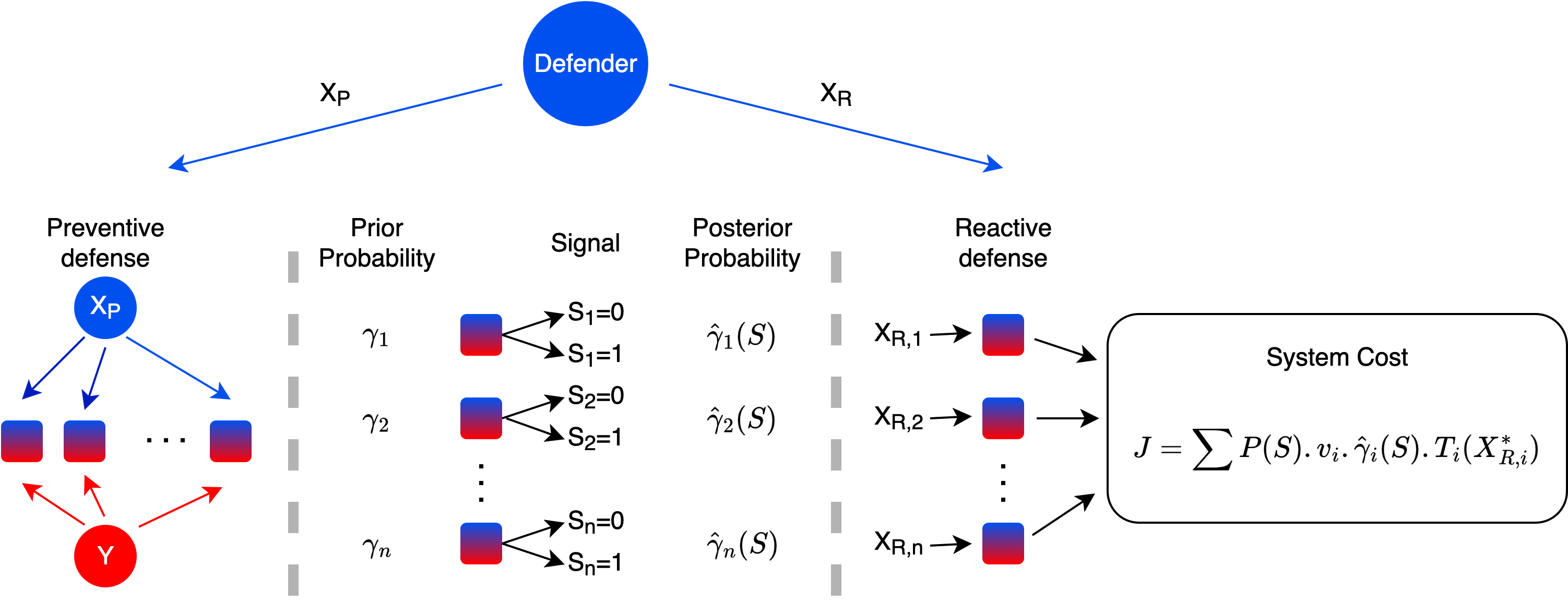}
    \caption{Illustration of the two-stage model.  
    \textbf{Stage 1:} The defender allocates preventive resources $X_{P,i}$ across nodes before the attack, aiming to reduce compromise probabilities. The compromise status (i.e. compromised or not) of each node is determined by the probabilities $\gamma_i$. 
    \textbf{Stage 2:} After receiving noisy sensor signals $S_i$ regarding the nodes' status, the defender updates its beliefs on the nodes' compromise status via Bayesian inference and then allocates reactive resources $X_{R,i}$ in order to minimize the expected recovery cost, which is given by the total expected time each node spends in a compromised state.}

    \label{fig:system_diagram}
\end{figure*}

\subsubsection{Sensor Signaling}

Each node \(i\) is monitored by a binary sensor emitting a signal \(S_i \in \{0,1\}\). A signal \(S_i = 1\) suggests possible node compromise, whereas \(S_i = 0\) indicates the node is likely secure. Due to sensor imperfections, noise affects observations. Let \(p_i = \mathbb{P}(S_i = 1 \mid \text{compromised})\) be the true-positive rate, denoting the probability the sensor accurately identifies a compromised node. In contrast, \(q_i = \mathbb{P}(S_i = 1 \mid \text{safe})\) represents the false-positive rate, which describes the chance the sensor wrongly reports a safe node as compromised.

Given the prior probability \(\gamma_i\) that node \(i\) is actually compromised, the signal distribution can be expressed as a mixture:
\begin{align}
\mathbb{P}(S_i = 1) &= \gamma_i p_i + (1 - \gamma_i) q_i, \\
\mathbb{P}(S_i = 0) &= \gamma_i (1 - p_i) + (1 - \gamma_i)(1 - q_i).
\end{align}

\begin{figure}[!htbp]
    \centering
\includegraphics[width=0.75\linewidth]{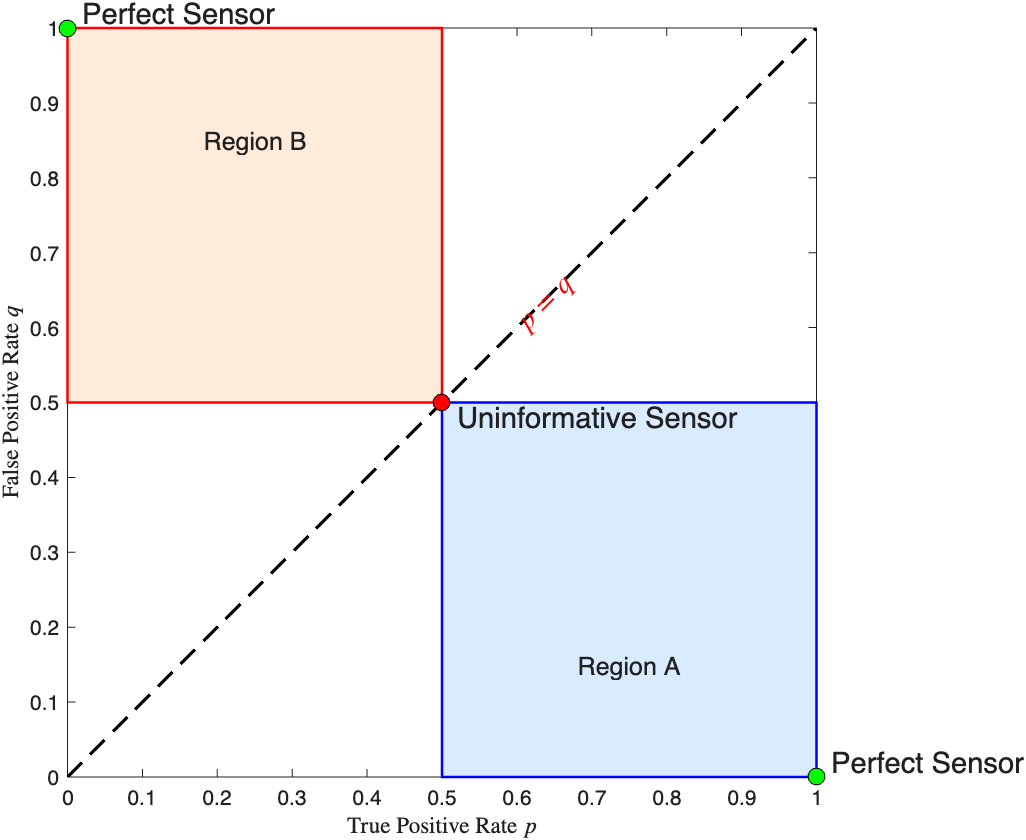}
    \caption{Sensor parameter space: we only need to  analyze the region induced by true-positive rate \(p \in [0.5,1]\) and false-positive rate \(q \in [0,0.5]\), because the diagonal \(p = q\) represents uninformative sensors and the region above this line is symmetric to the region below it via signal inversion.}
\label{fig:region_plot}
\end{figure}

As shown in Figure~\ref{fig:region_plot}, we restrict our analysis to the region:
\[
p_i \in [\frac{1}{2}, 1], \quad q_i \in [0, \frac{1}{2}]
\]
This region, denoted as \textit{Region A}, is justified as follows:

\begin{itemize}

    \item \textbf{\em Symmetry and equivalence:} The model exhibits an equivalence under the transformation \((p, q) \mapsto (1 - p, 1 - q)\), meaning that inverting the signal will flip its interpretation but preserves the information. Thus, studying the lower triangle \(p > q\) is sufficient to cover all cases without loss of generality.

    \item \textbf{\em Informativeness:} A sensor is informative only if \(p_i > q_i\), meaning that alarm is triggered when the node is compromised than when the node is secure; note that 
    sensors whose parameters are near the diagonal line \(p = q\) behave like random coin flipping and offer no value.

\ignore{   
    For example:
    \begin{itemize}
        \item \((p, q) = (1, 0)\): perfect sensor,
        \item \((p, q) = (\frac{1}{2}, \frac{1}{2})\): uninformative (no signal value),
        \item \((p, q) = (0, 1)\): misleading, but equivalent to flipping interpretation of 1 and 0.
    \end{itemize}
}
    
\end{itemize}
Thus, our investigation will focus on $Region A$.

\subsubsection{Bayesian Update}
\label{sec:bayesupdate}

Upon observing signal \(S_i\), the defender updates their belief on whether node $i$ is compromised by using Bayes’ rule \cite{probability}:
{\small
\begin{align}
\gamma_i(S_i) 
&= \frac{p_i^{S_i}(1 - p_i)^{1 - S_i} \cdot \gamma_i}
         {p_i^{S_i}(1 - p_i)^{1 - S_i} \cdot \gamma_i + q_i^{S_i}(1 - q_i)^{1 - S_i} \cdot (1 - \gamma_i)}.
\end{align}
}

\subsubsection{Reactive Defense Budget and Recovery Time}
\label{sec:reactivecost}

With probability $\hat{\gamma}_i(S_i)$, The defender will allocate reactive defense resources \(X_{R,i}\) to accelerate the recovery of node $i$ in the event of compromise.
We assume that the time required to recover a compromised node follows an exponential distribution, 
which is widely used in epidemic models and game-theoretic models (e.g. \cite{Van2013}).
Specifically, for node \(i\), given that it is truly compromised, we model the recovery time \(T_i\) as an exponential random variable with rate parameter $r_i(X_{R,i})$ given by:
\begin{align}
T_i &\sim \text{Exp}\left( r_i(X_{R,i}) \right), \notag \\
\text{where} \quad r_i(X_{R,i}) &= \frac{X_{R,i} + \delta_i}{X_{R,i} + \delta_i + \epsilon_i},
\label{eq:recovery_rate}
\end{align}
where \(\delta_i \ge 0\) is pre-exsisting defenses resources and \(\epsilon_i > 0\) is a half-saturation constant that governs how fast the recovery time decreases with the additional allocation of reactive resources. Thus, the expected recovery time is:
\begin{align}
T_i(X_{R,i}) = 1 + \frac{\epsilon_i}{\delta_i + X_{R,i}}.
\label{eq:Ti_cost}
\end{align}
Otherwise, if node \(i\) is not compromised, the recovery time is zero:
\[
    T_i = 0.
\]
This function is strictly decreasing and convex in \(X_{R,i}\), while reflecting diminishing returns on reactive defense budget.

\subsubsection{Defense Budget Allocation and the  Optimization Problem}
\label{sec:defenderopt}

Let \( X_P = (X_{P,1}, \dots, X_{P,n}) \) and \( XR = (X_{R,1}, \dots, X_{R,n}) \) denote the defender’s preventive and reactive defense budget allocation across the \(n\) nodes. For a given signal realization \( S = (S_1, \dots, S_n) \), the defender updates the posterior beliefs \( \hat{\gamma}_i(S_i) \) as described above, 
incurring the following cost:
\begin{equation}
J(S, X_P, X_R, p, q) := \sum_{i=1}^n v_i\, \hat{\gamma}_i(S_i)\, T_i(X_{R,i}).
\end{equation}
The total cost, averaged over all signal realizations, is:
\begin{equation}
J(X_P, X_R, p, q) := \sum_{S \in \{0,1\}^n} \mathbb{P}(S)\, J(S, X_P, X_R, p, q),
\end{equation}
where $\mathbb{P}(S)$ is the signal distribution:
\begin{align}
\mathbb{P}(S) := \prod_{i=1}^n 
& \left[ \gamma_i p_i + (1 - \gamma_i) q_i \right]^{S_i} \nonumber \\
& \times \left[ \gamma_i (1 - p_i) + (1 - \gamma_i)(1 - q_i) \right]^{1 - S_i}
\label{eq:signal_prob}
\end{align}
The optimization problem is for the defender to minimize this total cost subject to the resource constraint:
\begin{equation}
\min_{\substack{X_{P,i} \ge 0,\, X_{R,i} \ge 0 \\ \sum_{i=1}^n X_{P,i} + \sum_{i=1}^n X_{R,i} = X}} J(X_P, X_R, p, q).
\end{equation}
We observe that it is difficult to solve the optimization problem because the optimal solution
depends on the uncertainty imposed by the signals. 
We propose tackling the problem by leveraging the following {\em backward induction} strategy, in two steps.

\ignore{
\noindent{\bf Objective function.}
The objective of the study is to optimally allocate defense budget such that the expected time during which the nodes are compromised is minimal, under constraint
\begin{align}
\sum_{i=1}^{n} X_{P,i} + \sum_{i=1}^{n} X_{R,i} = X.
\end{align}
}

\noindent
\textbf{Step 1:  Reactive Defense Budget Allocation Optimization.}
Given preventive defense budget allocation \( X_P \) and a signal realization \( S \), the defender selects the reactive defense budget allocation \( X_R \) to minimize the recovery cost:
\begin{equation}
X_R^* := \arg\min_{\substack{X_{R,i} \ge 0 \\ \sum X_{R,i} = X - \sum X_{P,i}}} J(S, X_P, X_R, p, q)
\label{eq:stage2}
\end{equation}
This is a convex optimization problem that admits a closed-form solution via the Karush–Kuhn–Tucker (KKT) conditions, as given by Proposition~\ref{prop:kkt}.

\begin{proposition}[Closed-Form Optimal Reactive Defense Budget Allocation]
\label{prop:kkt}
Let \( X_P \in \mathbb{R}_+^n \) be fixed and \( S = (S_1, \dots, S_n) \in \{0,1\}^n \) be the observed signal vector. The optimal reactive defense budget allocation that solves Eq.~\eqref{eq:stage2} is given by:
{\small
\begin{align}
X_{R,i}^* := \max\left\{
\frac{\sqrt{v_i \hat{\gamma}_i(S_i) \epsilon_i}}{\sum_{k \in \mathcal{A}} \sqrt{v_k \hat{\gamma}_k(S_k) \epsilon_k}}
\left(X_R + \sum_{k \in \mathcal{A}} \delta_k\right) - \delta_i,\, 0
\right\},
\end{align}
}
where \( X_R = X - \sum_{i=1}^n X_{P,i} \) is the reactive defense budget and \( \mathcal{A} = \{ i : X_{R_,i}^* > 0 \} \) is the set of actively recovered nodes.
\end{proposition}

\noindent
\textbf{Step 2: Preventive Defense Budget Allocation Optimization.}
The expected cost under optimal reactive defense budget allocation is:
\begin{equation}
J^*(X_P; p, q) := \sum_{S \in \{0,1\}^n} \mathbb{P}(S)\, J(S, X_P, X_R^*, p, q).
\end{equation}
The defender selects the preventive defense budget allocation \( X_P \) to minimize the total expected cost via:
\begin{equation}
X_P^* = \arg\min_{\substack{X_{P,i} \ge 0 \\ \sum X_{P,i} \le X}} J^*(X_P; p, q).
\label{eq:preventivecost}
\end{equation}
The resulting minimum total cost is denoted by:
\begin{equation}
J^*(p, q) := J^*(X_P^*; p, q)
\label{eq:totalCost}
\end{equation}

\vspace{0.5em}
\noindent
This two-stage optimization structure defines the core logic and analysis of our model. Note that Step 1 
is analytically tractable via Proposition~\ref{prop:kkt}, while Step 2
is nonconvex in general and thus can only be solved numerically.

\begin{table}[!htbp]
\centering
\caption{Model parameters}
\begin{tabularx}{\linewidth}{@{}lX@{}}
\hline
\(n\)                    & number of nodes, indexed by \(i = 1, \dots, n\) \\
\(Y_i > 0\)              & attacker effort on attacking node \(i\) \\
\(X_{P,i} \ge 0\)           & defender's preventive defense budget spent on hardening node \(i\) \\
\(X_{R,i} \ge 0\)           & defender's reactive defense budget spent on cleaning up node \(i\) \\
\(p_i, q_i \in [0,1]\)   & true-positive rate and false-positive rate of the sensor employed at node \(i\), with \(p_i > q_i\)  \\
\(\gamma_i\)             & prior probability that node $i$ is compromised despite the employment of preventive defense budget $X_{P,i}$  \\
\(S_i \in \{0,1\}\)      & binary signal by the sensor employed at node \(i\), where \(S_i=1\) suggests the node may be compromised, \(S_i=0\) suggests it may be safe. \\
\(\hat{\gamma}_i(S_i)\)        & posterior probability that node $i$ is compromised given signal \(S_i\) \\
\(\epsilon_i > 0\)       & half-saturation constant governing the marginal benefit of reactive defense \\
\(\delta_i \ge 0\)       & pre-existing defense resources \\
\(T_i(X_{R,i})\)            & expected recovery time at node \(i\)\\
\(v_i > 0\)              & valuation of node \(i\) measured as cost per time \\
\(P(S)\)                 & joint probability of signal vector \(S = (S_1, \dots, S_n)\) \\
\(J(S)\)               & recovery cost under signal realization \(S\) \\
\(J\)                  & expected recovery cost, averaged over all signal vectors \\
\hline
\end{tabularx}
\label{tab:notation}
\end{table}

\section{Main Result}
\label{sec:result}
\subsection{Tradeoff Between Preventive and Reactive Allocation}

First, we consider the problem of determining the optimal division of the defense budget between preventive and reactive strategy \( X_{P,i} \quad and \quad X_{R,i} \) given a signal with parameters \( p \) (the true-positive rate) and \( q \) (the false-positive rate).Intuitively, better quality sensors produce signals that are more informative, allowing the defender to estimate the compromise state more accurately. The following theorem formalizes this monotonic relationship between sensor quality and preventive allocation.

\begin{theorem}[Monotonicity of Preventive Effort in \( p \) and \( q \)]
\label{thm:xp_total_vs_pq_n_nodes}
Consider the attacker–defender game with \( n \) nodes, where each node \( i \) has:
\(Y_i > 0 \quad X_{P,i} \ge 0\ \quad \gamma_i > 0 \)

Let \( X_{P,i}^* = (X_{P,1}^*, \dots, X_{P,n}^*) \) denote the optimal preventive allocation that solves Eq.~\eqref{eq:preventivecost}.
Then, for each node \( i \), the optimal preventive allocation satisfies:
\begin{align}
    \frac{d}{dp} X_{P,i}^* \ge 0, \quad \frac{d}{dq} X_{P,i}^* \le 0.
\end{align}

\end{theorem}

\begin{proof}
We do not assume a closed-form solution for \( X_{P,i}^* \). Instead, the structural relationships in the model are analyzed.

\textbf{Step 1: Posterior monotonicity.}  
For fixed \( \gamma_i \in (0,1) \), the posterior given \( S_i = 1 \) is:
\begin{align}
\hat{\gamma}_i(1) = \frac{p \cdot \gamma_i}{p \cdot \gamma_i + q \cdot (1 - \gamma_i)}.
\end{align}
Taking the derivative:
\begin{align}
\frac{d}{dp} \hat{\gamma}_i(1) > 0, \quad \frac{d}{dq} \hat{\gamma}_i(1) < 0.
\end{align}
Thus, \( \hat{\gamma}_i(1) \) becomes more reliable with higher \( p \) and less reliable with higher \( q \).

\textbf{Step 2: Defender incentive to reduce compromise probability.}  
The cost term for each node is:
\begin{align}
v_i \cdot \hat{\gamma}_i(S_i) \cdot T_i(X_{R,i}),
\end{align}
where \( T_i \) exhibits convex properties and diminishes as a function of \( X_{R,i} \), with \( X_{R,i} \) being adaptively determined based on the observed signals. This cost increases relative to \( \hat{\gamma}_i(S_i) \), thereby incentivizing the defender to minimize the initial \( \gamma_i \) and its responsiveness through strategic measures.

\textbf{Step 3: Inverse relation between \(X_P\) and \( \gamma_i \).}  
From:
\begin{align}
\gamma_i = \frac{Y_i}{Y_i + X_{P,i}} \quad \Leftrightarrow \quad X_{P,i} = \frac{Y_i (1 - \gamma_i)}{\gamma_i},
\end{align}
we conclude that \( X_{P,i} \) is strictly decreasing in \( \gamma_i \).

\textbf{Step 4: Final monotonicity result.}  
As \( p \uparrow \Rightarrow \hat{\gamma}_i(1) \uparrow \Rightarrow \) defender lowers \( \gamma_i \Rightarrow X_{P,i} \uparrow \).  
As \( q \uparrow \Rightarrow \hat{\gamma}_i(1) \downarrow \Rightarrow \) defender tolerates higher \( \gamma_i \Rightarrow XP_i \downarrow \).

Therefore, \( X_{P,i}^* \) increases in \( p \), decreasing in \( q \), which completes the proof.
\end{proof}

\begin{corollary}[Perfect Sensors Maximize Preventive Allocation]
\label{cor:perfect_sensors}
Under the same model, the total preventive allocation is maximized when \( p = 1 \), \( q = 0 \):
\begin{align}
\sum_{i=1}^n X_{P,i}^*(p = 1, q = 0) = \max_{p, q \in (0,1)} \sum_{i=1}^n X_{P,i}^*(p, q).
\end{align}
\end{corollary}

\begin{proof}
This follows directly from the monotonicity of \(X_P\) in \( p \) and \( q \).
\end{proof}

Theorem~\ref{thm:xp_total_vs_pq_n_nodes} shows that as the sensor becomes more accurate---through a higher true-positive rate $p$ or a lower false-positive rate $q$---the defender allocates more resources to prevention. This confirms our central hypothesis. This increases the effectiveness of reactive allocation and incentivizes the defender to allocate more budget toward prevention, knowing that sensor feedback will support precise post-attack response. In contrast, poor sensor quality leads to weaker reliability of reactive decisions, provoking the defender to allocate less resources to prevention and reserve more resources for generic reactive responses.

Significantly, this rise in preventive measures is evident both overall and at the individual node level. Enhanced sensor accuracy leads to an increased optimal allocation for all nodes, indicating a strategic realignment.

In order to empirically validate this result, we numerically solve the optimal preventive allocation \(X_P^*\) for a grid of values \(p,q\) and a fixed total budget \(X = 10\), using a numerical optimization method. The total preventive effort can be seen in Figures~\ref{fig:xp_vs_p} and~\ref{fig:xp_vs_q} to increase with \(p\) and decrease with \(q\), confirming the monotonicity of the optimal allocation predicted by the theorem.

\begin{figure}[]
\centerline{\includegraphics[width=0.47\textwidth]{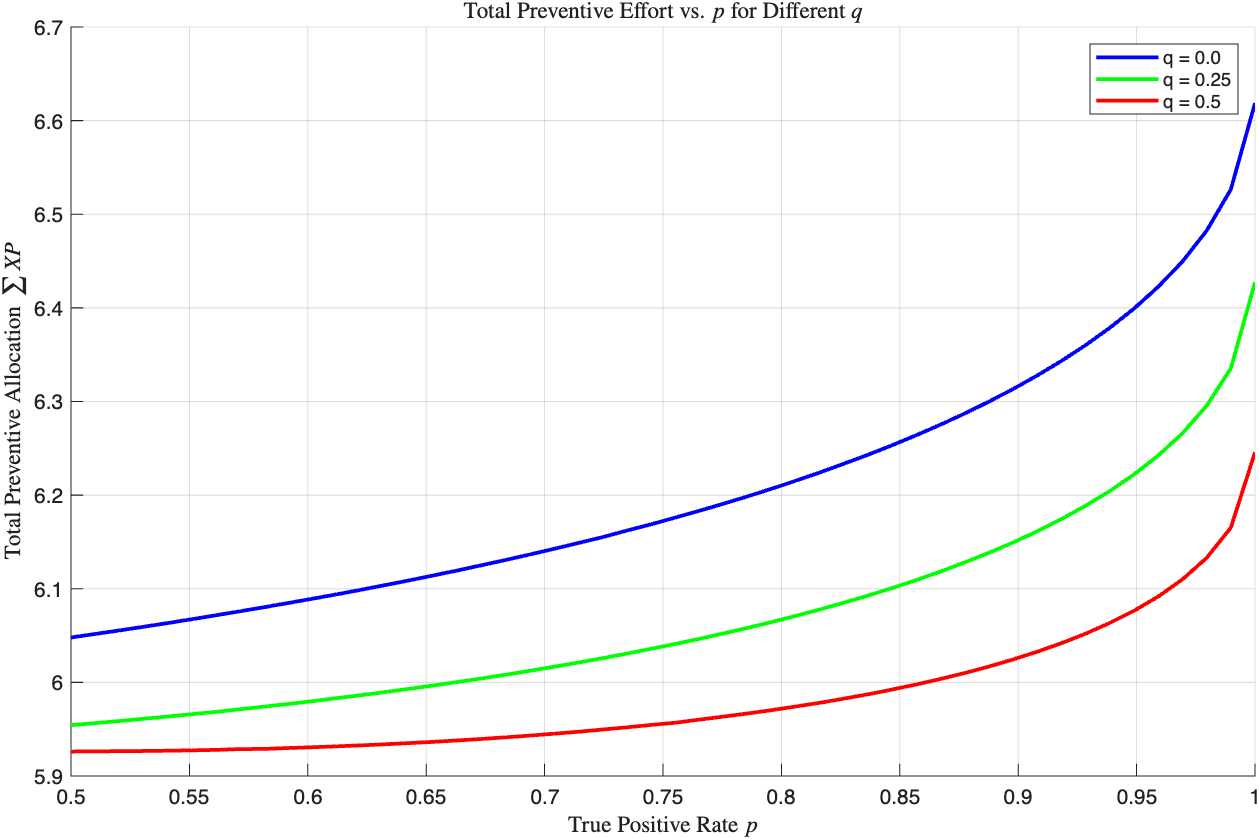}}
\caption{Total preventive allocation \( \sum X_{P,i} \) vs. true-positive rate \( p \), for fixed false-positive rates \( q \in \{0.0, 0.25, 0.5\} \).}
\label{fig:xp_vs_p}
\end{figure}

\begin{figure}[]
\centerline{\includegraphics[width=0.47\textwidth]{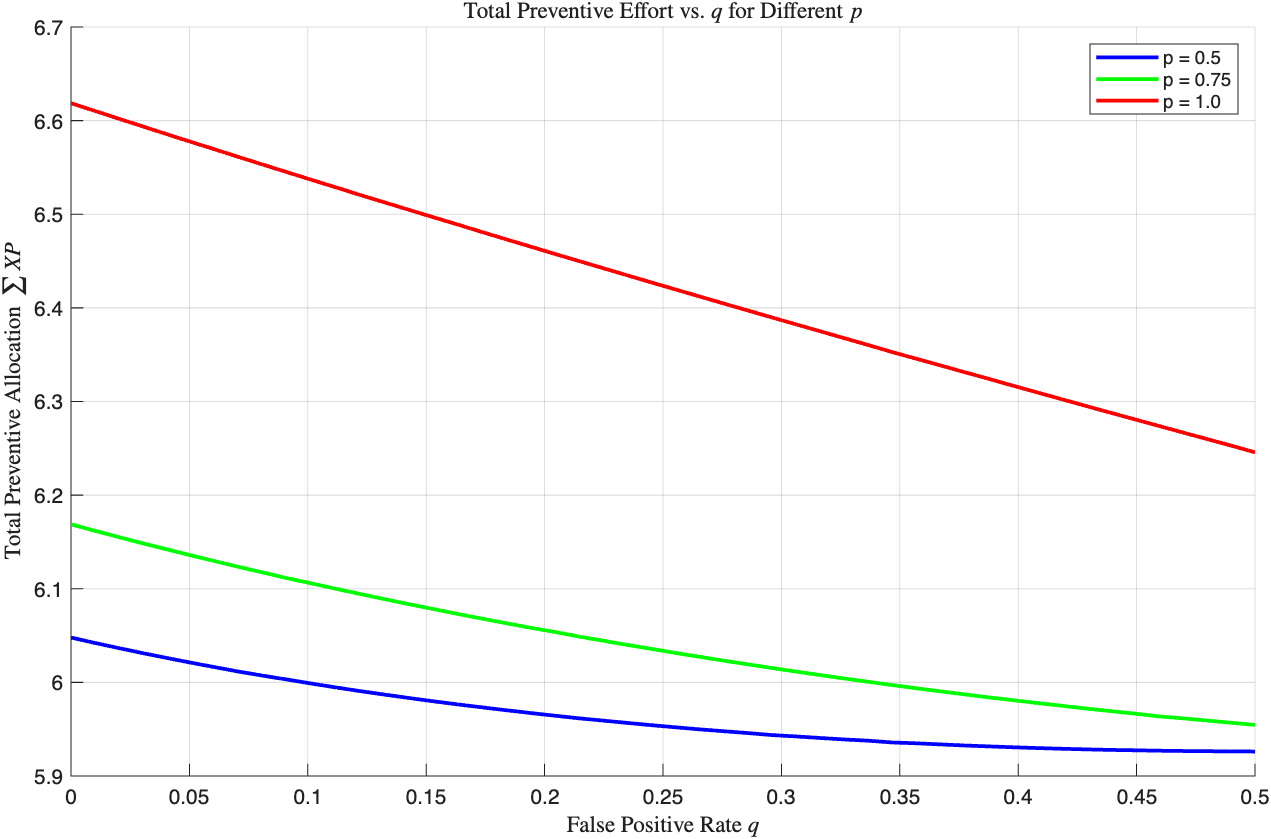}}
\caption{Total preventive allocation \( \sum X_{P,i} \) vs. false-positive rate \( q \), for fixed true-positive rates \( p \in \{0.5, 0.75, 1.0\} \).}
\label{fig:xp_vs_q}
\end{figure}

\subsection{Improvement from Sensor Information}

We now analyze how the presence of imperfect sensor signals affects the expected recovery cost of defense. 

To evaluate the value of sensor information, we compare the expected cost of the defender in two different settings:

\begin{itemize}
    \item \textbf{No-Sensor Baseline:} When the sensor does not provide information on node compromise, that is, \( p = q = \frac{1}{2} \). In this case, the posterior equals the prior, and the reactive allocation is based solely on prior compromise probabilities \( \gamma_i \). We denote the corresponding cost by \( J(\frac{1}{2}, \frac{1}{2}) \).
    
    \item \textbf{Sensor-Aided Case:} When the sensor is informative, that is, \( p \ne q \), and the defender updates posterior beliefs \( \hat{\gamma}_i(S_i) \) using the Bayes rule based on observed signals. This results in adaptive reactive allocations and a reduced expected cost denoted by \( J(p, q) \).
\end{itemize}

We define the value of sensing as the percentage reduction in expected cost as compared to the baseline without sensor. 

To quantify the benefit of sensing, we define the improvement function as:
\begin{align}
I(p,q) = 1 - \frac{J(p, q)}{J(\frac{1}{2}, \frac{1}{2})}.
\label{eq:impr}
\end{align}
This metric \(I(p, q)\) measures the relative reduction in expected cost due to sensor information. Specifically, \(I(p, q) = 0\) indicates no improvement (i.e., the sensor is uninformative), while \(I(p, q) = 1\) represents complete elimination of expected cost (i.e., perfect defense enabled by ideal sensor signals). Note that this is not a multiplicative improvement; rather, it reflects a fractional decrease in expected cost relative to the no-sensor baseline.

The following theorem characterizes important structural properties of the value of sensing and its relationship with the signal parameters p and q.

\begin{theorem}
\label{thm:improvement_sensor}
The improvement function \(I(p, q)\), defined as Eq. \ref{eq:impr}
satisfies the following structural properties:
\begin{enumerate}
    \item[\textbf{(i)}] \textbf{Zero Improvement when \( p = q \):}  
    \begin{align}
    I(p, p) = 0.
    \end{align}
    \item[\textbf{(ii)}] \textbf{Improvement Increases with \( p \) (for fixed \( q = 0 \)):}  
    \begin{align}
    \frac{dI}{dp}(p, 0) > 0.
    \end{align}
    \item[\textbf{(iii)}] \textbf{Improvement Decreases with \( q \) (for fixed \( p = 1 \)):}  
    \begin{align}
    \frac{dI}{dq}(1, q) < 0.
    \end{align}
\end{enumerate}
\end{theorem}

\begin{proof}
The definition of \( I(p, q) \) follows directly from the expected cost with sensor information, \( J(p, q) \), and the no-sensor baseline, \( J(\frac{1}{2}, \frac{1}{2}) \), both computed using the KKT-optimal reactive allocations.

\textbf{(i)} When \( p = q \), Bayes’ rule yields:
\begin{align}
\gamma_i(S_i) = \frac{p \cdot \gamma_i}{p \cdot \gamma_i + p \cdot (1 - \gamma_i)} = \gamma_i,
\end{align}
so the posterior equals the prior. The defender’s behavior is thus identical in both cases, and we have \( J(p, p) = J(\frac{1}{2}, \frac{1}{2}) \), hence \( I(p, p) = 0 \).

\textbf{(ii)} Fix \( q = 0 \). As \( p \) increases, the signal becomes more informative by reducing false negatives. In particular, when \( S_i = 0 \), the posterior \( \hat{\gamma}_i(0) \) becomes smaller, allowing the defender to allocate the reactive defense resources more accurately. This 
reduces the expected cost \( J(p, 0) \), so the improvement \( I(p, 0) \) increases. Therefore, \( \frac{dI}{dp}(p, 0) > 0 \).

\textbf{(iii)} Fix \( p = 1 \). As \( q \) increases, false-positives become more frequent, effectively reducing the reliability of sensor signals. This degrades the accuracy of the posterior part and causes misallocated reactive defense resources, raising the expected cost \( J(1, q) \). Hence, the improvement \( I(1, q) \) decreases, and \( \frac{dI}{dq}(1, q) < 0 \).
\end{proof}

Theorem~\ref{thm:improvement_sensor} states that the expected cost of the defender is decreasing in the provision of sensor information if and only if the sensor is informative, i.e., if \( p \ne q \). Furthermore, the improvement function increases monotonically in the true-positive rate \( p \) and decreases monotonically in the false-positive rate \( q \). This proves a key implication of our model: that even noisy sensing can lead to more effective reactive responses by updating posteriors. The monotonicity properties also indicate that modest improvements in detection performance can lead to meaningful cost savings. These intuitions are borne out in Fig.~\ref{fig:improvement_pq}, which plots the improvement function \( I(p,q) \) across the sensor parameter space.

Importantly, these performance improvements do not come from better reactivity alone. As shown in Theorem~\ref{thm:xp_total_vs_pq_n_nodes}, increasing \( p \) also increases the optimal preventive effort \( X_P^* \). Since reactive allocation in Stage 2 always uses the same KKT-based optimization, cost reduction \( J(p,q) \) must be mainly due to lower risk of compromise - caused by greater investment in prevention. Thus, better sensor information improves performance not by enabling a more precise reaction but by inducing a more confident and targeted prevention. This is a somewhat counter-intuitive outcome: sensing helps not because it reacts better, but because it motivates \emph{better preparation} before the attack.

\begin{figure}[]
\centerline{\includegraphics[width=0.47\textwidth]{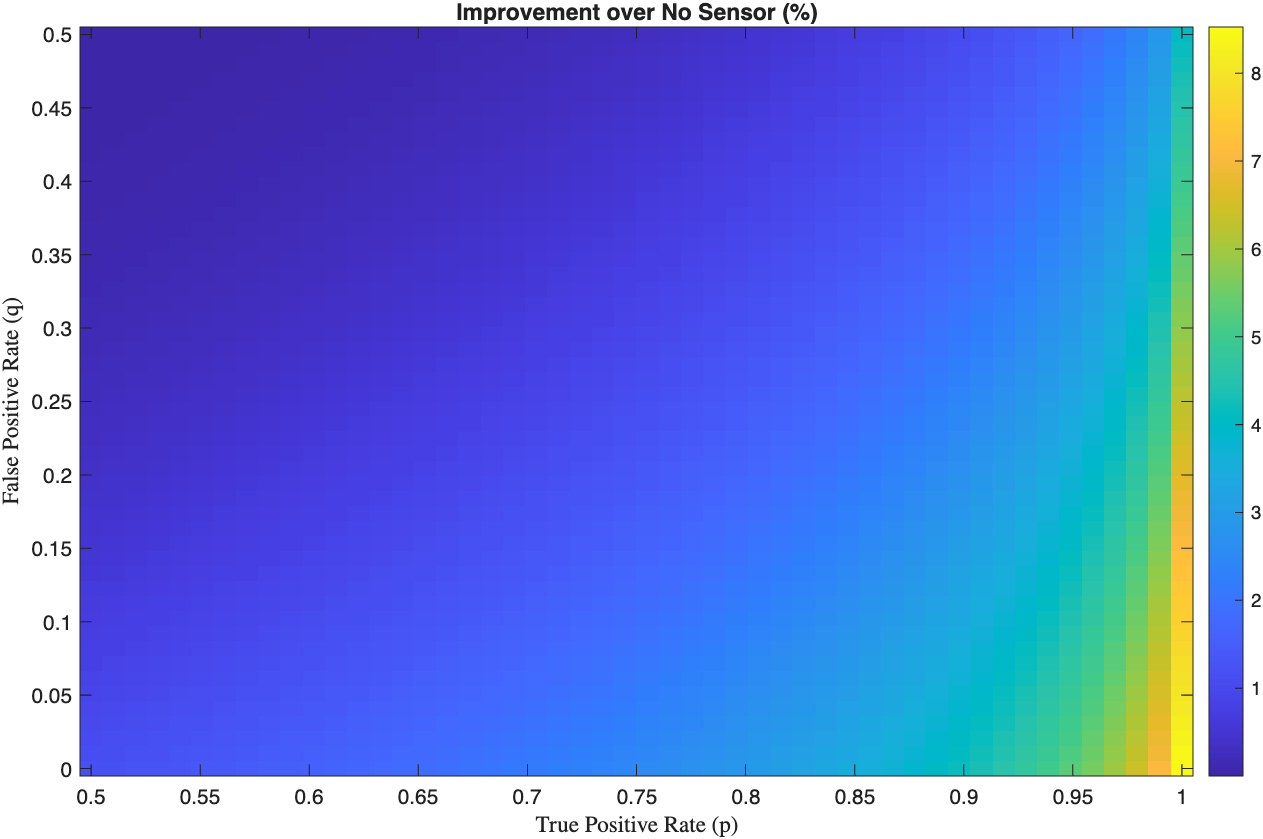}}
\caption{Percentage improvement \( I(p,q) \) over the signal parameter space.}
\label{fig:improvement_pq}
\end{figure}

\subsection{Improvement vs. Prior Compromise Probability}

Based on the previous result, where we showed that sensor quality \((p, q)\) influences the value of sensing, we now analyze how the defender's \emph{initial uncertainty} - captured by prior compromise probabilities \( \gamma_i \) - affects the marginal benefit of sensor information. In many real-world networks, these priors may correspond to previous risk assessment, node criticality/importance, or vulnerability to malware infections. The central question then is how to optimally address \emph{initial uncertainty} about node compromise, given the signal quality, to maximize the marginal benefit of sensor investment. In order to examine this factor in isolation, the uniform setting is used as described in previous paper, with fixed parameters except \(\gamma_i\).

Note that in previous sections, the expected cost was denoted \( J^*(p, q) \), where the prior compromise probabilities \( \gamma = (\gamma_1, \dots, \gamma_n) \) were held fixed. Here, we explicitly write \( J^*(p, q; \gamma) \) to emphasize that the cost depends on the initial uncertainty of the defender, which we vary to study its effect on the sensing value.

\vspace{0.5em}
\noindent\textbf{Baseline Cost without Sensor Information.}
Under uninformative signals (\(p = q = \frac{1}{2}\)), the defender allocates reactive resources based on prior probabilities only:
\begin{align}
J_n = J^*(\frac{1}{2}, \frac{1}{2}; \gamma) = v \sum_{i=1}^n \gamma_i \left(1 + \frac{\epsilon}{\delta + X_{R,i}^*} \right),
\end{align}
where \( J^*(p, q; \gamma) \) denotes the fully optimized cost (see Eq.~\eqref{eq:totalCost}) and \( X_{R,i}^* \) is the reactive allocation based on prior belief.

\vspace{0.5em}
\noindent\textbf{Cost under Perfect Sensor Information.}
Under perfect sensing (\(p = 1, q = 0\)), the reactive effort is allocated exactly to compromised nodes:
\begin{align}
    J_p = J(1,0,\gamma_i) = \sum_{S \subseteq [n]} \left( \prod_{i \in S} \gamma_i \prod_{j \notin S} (1 - \gamma_j) \right) \notag \\
\cdot |S|v \left(1 + \frac{\epsilon}{\delta + X_R/|S|} \right).
\end{align}
The expression for \( J_p \) accounts for all possible combinations of which nodes are actually compromised. Each subset \( S \subseteq [n] \) represents a possible realization of the compromised nodes, where node \( i \in S \) is compromised and \( j \notin S \) is not. The term
\[
\prod_{i \in S} \gamma_i \prod_{j \notin S} (1 - \gamma_j)
\]
is the joint probability that this particular subset \( S \) is the true compromise pattern, assuming independent priors. Given perfect sensing, the defender knows \( S \) exactly and allocates all reactive resources \( X_R \) to the compromised nodes in \( S \), distributing them equally. The total cost for a realization \( S \) is proportional to the number of compromised nodes, \( |S| \), each incurring a recovery delay cost
\[
1 + \frac{\epsilon}{\delta + X_R / |S|}.
\]
The overall expected cost \( J_p \) is then computed by averaging this cost over all possible subsets \( S \), weighted by their respective probabilities.
The improvement is defined as the fractional reduction in expected cost:
\[
I^{(n)}(\gamma_1, \dots, \gamma_n) := 1 - \frac{J_p}{J_n}.
\]

We now analyze how \( I^{(n)} \) varies as a function of a single prior, \( \gamma_1 \), with the other priors fixed.

\begin{theorem}[Unimodality and Endpoint Behavior of Sensor Improvement]
\label{thm:improvement_gamma}
Let \(n \ge 2\), with identical node parameters \(v_i = v\), \(\epsilon_i = \epsilon\), \(\delta_i = \delta\), and fix \( \gamma_2, \dots, \gamma_n \in (0,1) \). Then the improvement function \( I^{(n)}(\gamma_1, \gamma_2, \dots, \gamma_n) \) satisfies:
\begin{itemize}
    \item \( I^{(n)}(0, 0, \dots, 0) = 0 \):\\
    There is no benefit from sensor information when all nodes are certainly safe (i.e., no uncertainty to resolve).
    
    \item \( \displaystyle \lim_{\gamma_1 \to 0^+} \frac{\partial I^{(n)}}{\partial \gamma_1} = +\infty \):\\
    The marginal value of sensing is extremely high when the defender is nearly certain that node 1 is safe.

    \item \( \displaystyle \left. \frac{\partial I^{(n)}}{\partial \gamma_1} \right|_{\gamma_1 = 1} < 0 \):\\
    When the defender is fully certain that node 1 is compromised, additional sensor information reduces value due to redundancy or misleading noise.

    \item \( \displaystyle \max_{\gamma_1 \in [0,1]} I^{(n)}(\gamma_1, \dots) = \lim_{\gamma_1 \to 0^+} I^{(n)} \):\\
    The maximum improvement is achieved in the limit as the defender's prior belief approaches full certainty that node 1 is safe.

    \item \( \displaystyle \frac{\partial I^{(n)}}{\partial \gamma_1} < 0 \quad \forall \gamma_1 \in (0,1] \):\\
    The improvement function is strictly decreasing in \(\gamma_1\), confirming that sensing is most useful under high uncertainty and becomes less valuable as confidence in compromise grows.
\end{itemize}
\end{theorem}

\begin{proof}
We differentiate \( I^{(n)} = 1 - J_p / J_n \), giving:
\begin{align}
\frac{\partial I^{(n)}}{\partial \gamma_1}
= -\frac{
\frac{\partial J_p}{\partial \gamma_1} \cdot J_n - J_p \cdot \frac{\partial J_n}{\partial \gamma_1}
}{(J_n)^2}.
\end{align}

At \( \gamma_1 = 0 \), only the subsets \( S \ni 1 \) contribute to \( \partial J_{\text{perf}}/\partial \gamma_1 \), and their weight derivatives are strictly positive. Meanwhile, the no-sensor term satisfies:
\begin{align}
\gamma_1 \left(1 + \frac{\epsilon}{\delta + X_{R,1}^*} \right) \approx \gamma_1 + A \sqrt{\gamma_1} + B \gamma_1,
\end{align}
yielding:
\begin{align}
\frac{d}{d\gamma_1} \left[ \cdot \right] = 1 + B + \frac{A}{2\sqrt{\gamma_1}} \to +\infty.
\end{align}
Thus:
\begin{align}
\lim_{\gamma_1 \to 0^+} \frac{\partial I^{(n)}}{\partial \gamma_1} = +\infty.
\end{align}

At \( \gamma_1 = 1 \), perfect-sensor cost decreases due to reduced uncertainty, while no-sensor cost increases smoothly:
\begin{align}
\left. \frac{\partial I^{(n)}}{\partial \gamma_1} \right|_{\gamma_1 = 1} < 0.
\end{align}

Since \( I^{(n)}(0) = 0 \) and the derivative is strictly negative for all \( \gamma_1 \in (0,1] \), the function reaches its maximum in the limit \( \gamma_1 \to 0^+ \), completing the proof.
\end{proof}

\vspace{0.5em}

Theorem~\ref{thm:improvement_gamma} reveals a somewhat counterintuitive insight: the value of sensor information is not maximized against highly likely threats, but rather against \emph{weak or low-probability threats} that the defender could otherwise overlook. As in \( \gamma_1 \to 1 \), the defender is already prepared to respond, so sensor feedback is no longer needed. As in \( \gamma_1 \to 0 \) As in the \( \gamma_1 \to 0 \) limit, the defender assumes that the node is safe and may be undersupplied, but in the process, the marginal value of the detection is small due to the low probability of compromise.

Empirically, Figure~\ref{fig:impr_gamma} shows that \( I^{(2)}(\gamma_1, \gamma_2) \) is maximized at \( \gamma_1 \approx 0.1 \) and drops on either side. In other words, sensing is most valuable when the threat is \emph{weak but plausible}, rather than when the defender is totally uncertain or totally confident.
This contrasts previous work that shows that information is most valuable against a stronger or more targeted opponent in adversarial settings, such as Colonel Blotto or General Lotto \cite{paarporn2022, Paarporn2025, TAC2024}. In our setting, information helps the defender in adjusting the defense posture toward otherwise neglected nodes, highlighting the strategic value of detection in early-stage or covert threats.

\begin{figure}[]
\centerline{\includegraphics[width=0.48\textwidth]{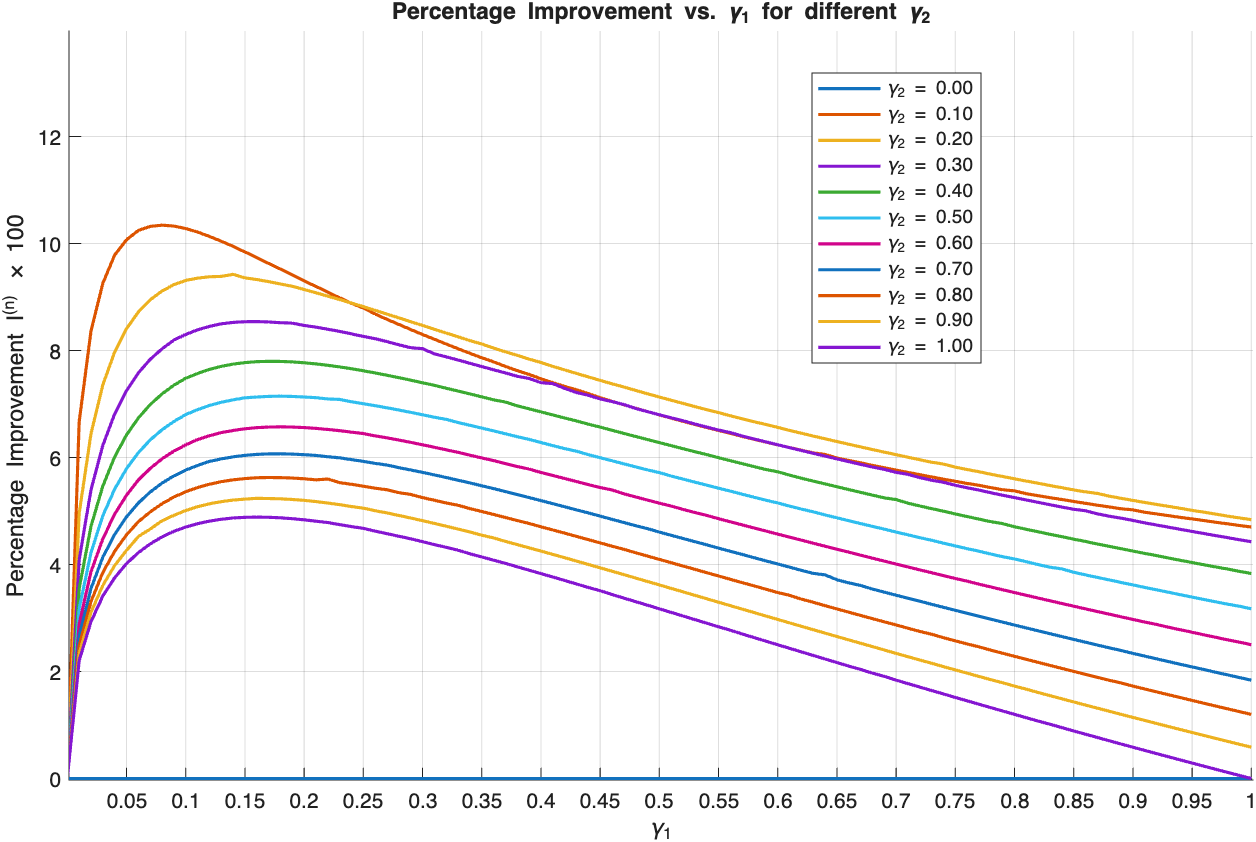}}
\caption{Relative improvement \( I^{(2)}(\gamma_1, \gamma_2) \) versus \( \gamma_1 \), for several values of \( \gamma_2 \). Sensor quality is fixed at \( (p = 1, q = 0) \).}
\label{fig:impr_gamma}
\end{figure}

\section{Conclusion}
\label{sec:con}
We developed a theoretical 
framework to model the tradeoff between prevention and response under sensor signal uncertainties. 
The defender performs Bayesian inference, based on sensor signals or observations, 
for informed defense resource allocation. We show that the optimal defense investment increases with the true-positive detection rate and decreases with the false-positive rate. We introduce an improvement function \(I(p,q)\) to quantify how sensor signal quality affects the defender's performance; \(I(p,q)\) increases with the true-positive rate \(p\) and decreases with the false-positive rate \(q\), highlighting the vital role of informative sensors in reducing the recovery costs.
We 
derive an improvement function \(I^{(n)}(\gamma_1,\dots,\gamma_n)\) for closed homogeneous systems, where \(I^{(n)}\) exhibits a unimodality in each \(\gamma_i\) while peaking at moderate prior uncertainty.  Numerical experiments confirm that informative sensors lower the expected recovery cost.

There are many exciting future research problems, such as the following: How should we deal with networked dependencies, where the security states of the nodes depend on each other?
How should we deal with strategic attackers that adapt to the defender's strategies?
How should 
the defender make real-time decisions using real-time sensor signals and online allocation of resources?
How should we optimize the placement of sensors?
How can we deal with the uncertainty associated with the prior probabilities and the uncertainties associated with sensor noise (i.e., their true-positive and false-positive rates)?

\bibliographystyle{IEEEtran}



\bibliographystyle{ieeetr}

\end{document}